\documentclass{llncs}
\usepackage{clrscode}
\usepackage{amsmath}
\usepackage{amssymb}
\usepackage{graphicx}
\usepackage{cite}
\newcommand{\ecc}{\varepsilon}

\begin{document}

\title{An optimal algorithm for the weighted backup 2-center problem on a tree}

\author{Hung-Lung Wang}
\institute{Institute of Information and Decision Sciences\\ National Taipei University of Business, Taiwan}
\maketitle

\begin{abstract}
In this paper, we are concerned with the weighted backup 2-center problem on a tree. The backup 2-center problem is a kind of center facility location problem, in which one is asked to deploy two facilities, with a given probability to fail, in a network. Given that the two facilities do not fail simultaneously, the goal is to find two locations, possibly on edges, that minimize the expected value of the maximum distance over all vertices to their closest functioning facility. In the weighted setting, each vertex in the network is associated with a nonnegative weight, and the distance from vertex $u$ to $v$ is weighted by the weight of $u$. With the strategy of prune-and-search, we propose a linear time algorithm, which is asymptotically optimal, to solve the weighted backup 2-center problem on a tree. 

\end{abstract}

\section{Introduction}

Facility location problems are widely investigated in the fields of operations research and theoretical computer science. The $p$-center problem is a classic one in this line of investigation. Given a graph $G$
with positive edge lengths, a  supply set $\Sigma$
, and a  demand set $\Delta$, the $p$-center problem asks for $p$ elements $x_1,x_2, \dots, x_p$ from $\Sigma$ such that $\max_{y\in\Delta}\min_{1\leqslant i\leqslant p}d(x_i, y)$ is minimized, where $d(x, y)$ denotes 
the distance from $x$ to $y$ in $G$. Conventionally, $\Delta, \Sigma\in \{V, A\}$, where $V$ and $A$ are the set of vertices and \textit{points} of $G$, respectively. A point of a graph is a location on an edge of the graph, and is identified with the edge it locates on and the distance to an end vertex of the edge. 
The $p$-center problem in general graphs, for arbitrary $p$, is NP-hard~\cite{KaHa79}, and the best possible approximation ratio is 2, unless NP=P~\cite{Ples87}. When $p$ is fixed or the network topology is specific, many efficient algorithms were proposed~\cite{BeBh06,BeBh07,Fred91}. 

\medskip

There are many generalized formulations of the center problem, like the capacitated center problem~\cite{AlDi10} and the minmax regret center problem~\cite{AvBe97,AvBe00}. The backup center problem is formulated based on the \textit{reliability model}~\cite{Snyd06,Snyd05}, in which the deployed facilities may sometimes fail, and the demands served by these facilities have to be reassigned to functioning facilities.  More precisely, in the backup $p$-center problem, facilities may fail with \textit{failure probabilities} $\rho_1, \rho_2, \dots, \rho_p$. Given that the facilities do not fail simultaneously, the goal is to find $p$ locations that minimize the expected value of the maximum distance over all vertices to their closest functioning facility. We leave the formal problem definition to Section~\ref{section:pre}. The backup $p$-center problem is NP-hard since it is a generalized formulation of the $p$-center problem. For $p=2$, Wang et al.~\cite{WaWu09} proposed a linear time algorithm for the problem on trees. When the edges are of identical length, Hong and Kang~\cite{HoKa12} proposed a linear time algorithm on interval graphs. Recently, Bhattacharya et al.~\cite{BhDe14} consider a weighted formulation of the backup 2-center problem, in which each vertex is associated with a nonnegative weight, and the distance from vertex $u$ to $v$ is weighted by the weight of $u$. They proposed $O(n)$-, $O(n\log n)$-, $O(n^2)$-, and $O(n^2\log n)$-time algorithms on paths, trees, cycles, and unicycles, respectively, where $n$ is the number of vertices.  

\medskip

In this paper, we focus on the weighted backup 2-center problem on a tree and design a linear time algorithm to solve this problem. The algorithm is asymptotically optimal, and therefore improves the current best result on trees, given by Bhattacharya et al.~\cite{BhDe14}. The strategy of our algorithm is \textit{prune-and-search}, which is widely applied in solving distance-related problems~\cite{Megi83,Tami04}. The rest of this paper is organized as follows. In Section~\ref{section:pre}, we formally define the problem and briefly review the result given by Bhattacharya et al.~\cite{BhDe14}. Based on their observations, a further elaboration on the objective function is given. In Section~\ref{section:linear}, we design the linear time algorithm, and concluding remarks are given in Section~\ref{section:conclusion}.

\section{Preliminaries}\label{section:pre}

Let $T=(V, E)$ be a tree, on which each vertex $v$ is associated with a nonnegative weight $w_v$,  and each edge is associated with a nonnegative length. A location on an edge is identified as a \textit{point}, and the set of points of $T$ is denoted by $A$. The unique path between two points $u$ and $v$ is denoted by $\pi(u, v)$, and the \textit{distance} $d(u, v)$ between two points $u$ and $v$ is defined to be the sum of lengths of the edges on $\pi(u, v)$. The \textit{weighted distance} from vertex $u$ to point $a$ is defined as $w_ud(u, a)$. The \textit{eccentricity} of a point $a$ is defined as 
\[\ecc(a)=\max_{u\in V}w_ud(u, a),\]
and the point with minimum eccentricity is said to be the \textit{weighted center} of $T$.
Note that the weighted center of a tree is unique. 
For $U\subseteq V$, the eccentricity of a vertex $a$ w.r.t. $U$ is defined as 
\[\ecc(a, U)=\max_{u\in U}w_ud(u, a).\]
Let $a_1$ and $a_2$ be two points of $T$. The partition $\Pi(a_1, a_2)$ of $V$ is defined as $(V_1, V_2)$, where $V_1=\{v\in V\colon\, d(v, a_1)\leqslant d(v, a_2)\}$ and $V_2=V\setminus V_1$. 
A \textit{weighted 2-center} consists of two points $c_1$ and $c_2$ minimizing 
\[\ecc_2(a_1, a_2)=\max\{\ecc_{}(a_1, V_1), \ecc_{}(a_2, V_2)\},\]
where $(V_1, V_2)=\Pi(a_1, a_2)$. We denote a weighted 2-center by $\{c_1, c_2\}$. Unlike the weighted center of a tree, there may be more than one weighted 2-center. 
Now we are ready to define the weighted backup 2-center problem. 

\medskip

\begin{problem}[the weighted backup 2-center problem]
Given a tree $T=(V, E)$ and two real numbers $\rho_1$ and $\rho_2$ in $[0, 1)$, the weighted backup 2-center problem asks for a point pair $(b_1, b_2)$ minimizing $\psi_{\rho_1, \rho_2}\colon\, A\times A\to \mathbb{R}$, where
\begin{eqnarray*}
\psi_{\rho_1, \rho_2}(a_1, a_2)&\equiv&(1-\rho_1)(1-\rho_2)\max\{\ecc_{}(a_1, V_1), \ecc_{}(a_2, V_2)\}\\
&&+\rho_2(1-\rho_1)\ecc(a_1)+\rho_1(1-\rho_2)\ecc(a_2),
\end{eqnarray*}
 and $(V_1, V_2)=\Pi(a_1, a_2)$. 
\end{problem}
\medskip
To ease the presentation, we assume that $\rho_1=\rho_2=\rho$. With the assumption, minimizing $\psi_{\rho_1, \rho_2}$ is equivalent to minimizing $\psi\colon\, A\times A\to \mathbb{R}$, where 
\[\psi(a_1, a_2)\equiv(1-\rho)\max\{\ecc_{}(a_1, V_1), \ecc_{}(a_2, V_2)\}+\rho(\ecc(a_1)+\ecc(a_2)).\]
We note here that all the proofs in this paper can immediately be extended to the case where failure probabilities are different. Moreover, $b_1$ and $b_2$ are not restricted to be deployed on different points. If $b_1$ and $b_2$ are identical,  the point must be the weighted center, as shown in Proposition~\ref{proposition:b2center_identical}.

\medskip

\begin{proposition}\label{proposition:b2center_identical}
Let $(b_1, b_2)$ be a weighted backup 2-center of tree $T$. If $b_1$ and $b_2$ are identical, then it is the weighted center of $T$.
\end{proposition}

\begin{proof}
Let $c$ be the weighted center of $T$. Suppose to the contrary that $b_1=b_2$, but $b_1\neq c$. Since $b_1=b_2$, we have $\{v\in V\colon\, d(v, b_1)\leqslant d(v, b_2)\}=V$, and therefore 
\begin{eqnarray*}
\psi(b_1, b_2)&=&(1-\rho)\max\{\ecc(b_1, V), \ecc(b_2, {\emptyset})\}+\rho(\ecc(b_1)+\ecc(b_2))\\
&=&(1-\rho)\ecc(b_1)+\rho(\ecc(b_1)+\ecc(b_2))\\
&>&(1-\rho)\ecc(c)+\rho(\ecc(c)+\ecc(c))\\
&=&\psi(c, c),
\end{eqnarray*}
which contradicts that $(b_1, b_2)$ is a weighted backup 2-center.\qed
\end{proof}

When computing a weighted backup 2-center, any vertex with weight zero can be treated as a point on an edge, and any edge $uv$ with length zero can be contracted to be a vertex with weight $\max\{w_u, w_v\}$. With  this manipulation, an instance with ``nonnegative constraints'' on vertex weights and edge lengths can be reduced to one with ``positive constraints'', and there is a straightforward correspondence between the solutions.
Therefore, in the discussion below, we may focus on the instances with positive vertex weights and edge lengths.

\subsection{A review on Bhattacharya's algorithm}\label{section:review}

\begin{figure}[t]
\centering
\includegraphics[width=2.2 in]{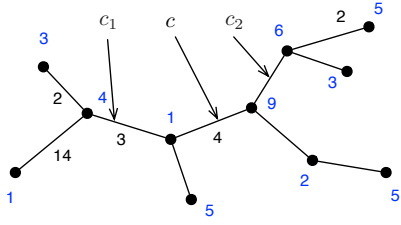}
\caption{A tree $T=(V, E)$. The number beside each vertex and each edge is the weight of the vertex and the length of the edge, respectively. Edges with no number aside are of length one.} \label{figure:tree}
\end{figure}
Throughout the rest of this paper, we use the tree given in Figure~\ref{figure:tree} as an illustrative example. In addition, weighted centers and weighted 2-centers are referred to as centers and 2-centers, respectively, for succinctness. The algorithm of Bhattacharya et al. depends on the following observations.

\begin{lemma}[See~\cite{BhDe14}]\label{lemma:solution_location}
Let	$\{c_1, c_2\}$ be any 2-center. There is a weighted backup 2-center $(b_1, b_2)$ such that $b_1$ (resp. $b_2$) lies on a path between $c_1$ (resp. $c_2$) and $c$.
\end{lemma}

\begin{lemma}[See~\cite{BhDe14}]\label{lemma:left_equal_right}
If $\rho > 0$, then $\ecc(b_1,V_1)= \ecc(b_2,V_2)$ holds for a weighted backup 2-center $(b_1, b_2)$ on a tree, where $(V_1, V_2)=\Pi(b_1, b_2)$.
\end{lemma}

By Lemma~\ref{lemma:solution_location}, we may focus the nontrivial case where $c_1<c<c_2$. The path $\pi(c_1, c_2)$  is embedded onto the $x$-axis with each point $a$ on $\pi(c_1, c_2)$ corresponding to point $d(a, c_1)$ on the $x$-axis. For simplicity, we use $\pi(c_1, c_2)$ to denote both the set of points on this path and the corresponding set of points on the $x$-axis.
For each vertex $v$, the \textit{cost function} $f_v(x)\colon\, \pi(c_1, c_2)\to \mathbb{R}$ is defined as 
\[f_v(x)=w_vd(v, x).\]

\medskip

\noindent Clearly, $f_v$ is a V-shape function. Assume that the minimum of $f_v$ occurs at~$a_v$. Let $f^+_v\colon\, \pi(c_1, c_2)\to \mathbb{R}$ and $f^-_v\colon\, \pi(c_1, c_2)\to \mathbb{R}$ be defined as 
\[f^+_v(x)=\left\{\begin{array}{ll}f_v(x),&\mbox{ if } x\geqslant a_v\\
0, &\mbox{ otherwise,}\end{array}\right. \mbox{ and } f^-_v(x)=\left\{\begin{array}{ll}f_v(x),&\mbox{ if } x\leqslant a_v\\
0, &\mbox{ otherwise,}\end{array}\right.\]
and let the upper envelopes of $\{f^+_v\colon\, v\in V\}$ and $\{f^-_v\colon\, v\in V\}$ be denoted by $E_L$ and $E_R$, respectively. An example is given in Figure~\ref{figure:envelopes_and_quasiconvex}(a). 
\begin{figure}[t]
\centering
\includegraphics[width=4.4 in]{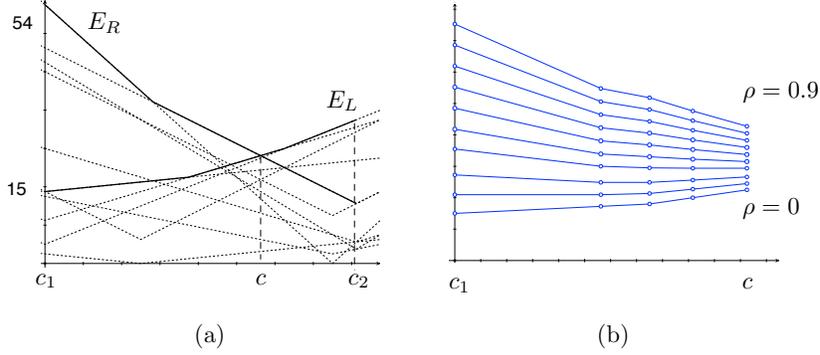}
\caption{(a) $E_L$ and $E_R$ of the tree $T$ given in Figure~\ref{figure:tree}. Dotted curves are functions $f_v$, and solid curves are $E_L$ and $E_R$. (b) Function $\psi_1$ with different failure probabilities $\rho$. The failure probability $\rho$ ranges from $0$ to $0.9$.} \label{figure:envelopes_and_quasiconvex}
\end{figure}

\bigskip

In the algorithm of Bhattacharya et al., they focus on processing the information within the path from $c_1$ to $c_2$, where $\{c_1, c_2\}$ is a 2-center satisfying the following property.
\begin{property}[See~\cite{BhDe14}]\label{property:center_2center}
For the center $c$, we have $\ecc(c)=E_L(c)=E_R(c)$. In addition, there is a 2-center $\{c_1, c_2\}$ satisfying $\ecc(c_1, V_1)=E_L(c_1)\leqslant \ecc(c)$ and $\ecc(c_2, V_2)=E_R(c_2)\leqslant \ecc(c)$, where $(V_1, V_2)=\Pi(c_1, c_2)$.
\end{property}
Property~\ref{property:center_2center} holds due to the continuity of the solution space, i.e., the set of points of $T$. 
Moreover, it can be derived from Property~\ref{property:center_2center} that for any point pair $(x_1, x_2)$ with $x_1\in \pi(c_1, c)$, $x_2\in\pi(c, c_2)$, and $E_L(x_1)=E_R(x_2)$, the partition $(V_1, V_2)=\Pi(x_1, x_2)$ satisfies $\ecc(x_1, V_1)=E_L(x_1)=E_R(x_2)=\ecc(x_2, V_2)$ since $\{c_1, c_2\}$ is a 2-center. 
As a result, Bhattacharya et al. gave an algorithm to compute a weighted backup 2-center on a tree $T$. We summarize it as in $\proc{BU2Center-Tree}$.

\medskip

\begin{codebox}
\Procname{$\proc{BU2Center-Tree}(T, \rho)$}
\li	$c\gets$ center of $T$
\li	$\{c_1, c_2\}\gets$ 2-center of $T$
\li	compute $E_L$ and $E_R$
\li	\For each bending point $x_1$
\li		\Do $x_2\gets x^*$, where $E_L(x_1)=E_R(x^*)$
\li			evaluate $\psi(x_1, x_2)$, and keep the minimum
	\End
\li	\Return $(x_1, x_2)$
\end{codebox}

\medskip

\noindent This algorithm runs in $O(n\log n)$ time, where $n$ is the number of vertices. The bottleneck is the computation of $E_L$ and $E_R$. Once $E_L$ and $E_R$ are computed, the remainder can be done in $O(n)$ time since there are $O(n)$ \textit{bending points}, at which the function $\psi(x, x^*)$ is not differentiable w.r.t. $x$. While processing $x_1$ from left to right, the corresponding $x_2$ moves monotonically to the left, and therefore a one-pass scan is sufficient to find the optimal solution. Readers can refer to~\cite{BhDe14} for details. To improve the time complexity, we elaborate on some properties of the objective function below.

\subsection{Properties}

As in~\cite{BhDe14}, the discussion below focuses on the behavior of the objective function on $\pi(c_1, c_2)$. We observe that the objective function possesses a good property (the quasiconvexity, given in Lemma~\ref{lemma:quasiconvex}) when $c_1$ and  $c_2$ satisfy certain restrictions. The existence of such a 2-center is proved in Proposition~\ref{proposition:two_center_equal}.

\begin{proposition}\label{proposition:two_center_equal}
For any tree $T=(V, E)$, there is a 2-center $\{c_1, c_2\}$ satisfying $\ecc(c_1, V_1)=\ecc(c_2, V_2)$, where $(V_1, V_2)=\Pi(c_1, c_2)$.
\end{proposition}

\begin{proof}

Let $\{c'_1, c'_2\}$ be a 2-center, and we embed $\pi(c'_1, c'_2)$ onto the $x$-axis as in Section~\ref{section:review}. Without loss of generality assume that $\ecc(c'_1, V'_1)<\ecc(c'_2, V'_2)$, where $(V'_1, V'_2)=\Pi(c'_1, c'_2)$, and it follows that $\ecc(c^{\prime}_2, V^{\prime}_2)=E_R(c^{\prime}_2)$ due to the continuity of $\pi(c^{\prime}_1, c^{\prime}_2)$. We claim that one can have the requested 2-center by moving $c^{\prime}_1$, towards $c^{\prime}_2$ along $\pi(c^{\prime}_1, c^{\prime}_2)$, to a point $a$ such that $E_L(a)=\ecc(c'_2, V'_2)$.   Let $(U_1, U_2)=\Pi(a, c'_2)$. Clearly, $V_1\subseteq U_1$ and $U_2\subseteq V_2$, and $\ecc(c^{\prime}_2, U_2)=\ecc(c^{\prime}_2, V^{\prime}_2)=E_R(c^{\prime}_2)$.

\medskip

Suppose to the contrary that $\{a, c'_2\}$ is not the requested 2-center. It follows that either (i) $\{a, c'_2\}$ is not a 2-center, or (ii) $\ecc(a, U_1)\neq\ecc(c'_2, U_2)$. For (i), there is a vertex $u$ in $U_1\setminus V_1$ satisfying 
\[\ecc(c'_2, U_2)<\ecc(a, U_1)=f_u(a)<f_u(c'_2)\leqslant\ecc(c'_2, V_2),\]
a contradiction. For (ii), it can be derived that $\ecc(a, U_1)<\ecc(c'_2, U_2)$, which contradicts the definition of $a$.
\qed
\end{proof}

\medskip

\begin{remark}
Any 2-center $\{c_1, c_2\}$ with $\ecc(c_1, V_1)=\ecc(c_2, V_2)$ satisfies $E_L(c_1)=\ecc(c_1, V_1)=\ecc(c_2, V_2)=E_R(c_2)$. Once a 2-center is computed, $\{c_1, c_2\}$ can then be computed in linear time, based on the arguments in the proof of Proposition~\ref{proposition:two_center_equal}. 
\end{remark}

\medskip

In the rest of this paper, we assume that $\{c_1, c_2\}$ is a 2-center satisfying $\ecc(c_1, V_1)=\ecc(c_2, V_2)$, where $(V_1, V_2)=\Pi(c_1, c_2)$. 
Next, we elaborate on $E_L$ and $E_R$. As noted in~\cite{BhDe14}, both $E_L$ and $E_R$ are piecewise linear.
On a path, both $E_L$ and $E_R$ are obviously continuous and convex. It also holds on $\pi(c_1, c_2)$ in a tree, as shown in Lemma~\ref{lemma:continuity}.

\begin{figure}[t]
\centering
\includegraphics[width=1.7 in]{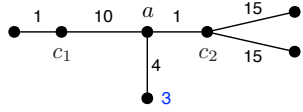}
\caption{An example which shows that $E_L$ is not continuous. All vertices are of weight $1$ except the bottom one, whose weight is $3$. Numbers beside edges are the edge lengths. Clearly, $E_L$ is not continuous at $a$.} \label{figure:discontinuity}
\end{figure}

\medskip

\begin{lemma}\label{lemma:continuity}
Let  $\{c_1, c_2\}$ be a 2-center of a tree satisfying $E_L(c_1)=\ecc(c_1, V_1)=\ecc(c_2, V_2)=E_R(c_2)$, where $(V_1, V_2)=\Pi(c_1, c_2)$. The function $E_L$ and $E_R$ are continuous and convex on $\pi(c_1, c_2)$. 
\end{lemma}

\begin{proof}
Because of symmetry, we prove the lemma only for $E_L$, and we claim that $E_L$ is continuous. The convexity then follows since $E_L$ is the upper envelope of half lines of positive slope. Suppose to the contrary that  $E_L$  is not continuous. There is a point $a$, with $c_1<a<c_2$, satisfying 
$\lim_{x\rightarrow a^-}E_L(x)<E_L(a).$
Let $E_L(a)=f^+_v(a)$. Clearly, at point $a$ we have
$f^+_v(a)=f^-_v(a).$
It follows that $a\notin V_1$ since otherwise $E_L(c_1)\neq\ecc(c_1, V_1)$. However, for $a\in V_2$, we have 
\[\ecc(c_2, V_2)=\ecc(c_1, V_1)=E_L(c_1)<E_L(a)=f^+_v(a)<f^+_v(c_2)\leqslant\ecc(c_2, V_2),\]
which leads to a contradiction.
\qed
\end{proof}
\medskip

\noindent Notice that Lemma~\ref{lemma:continuity} does not hold for all 2-centers. An example is given in Figure~\ref{figure:discontinuity}.

\bigskip

By Lemma~\ref{lemma:left_equal_right}, an optimal solution occurs at the point pair $(a, a^*)$ satisfying $a\in\pi(c_1, c)$, $a^*\in\pi(c, c_2)$, and $E_L(a)=E_R(a^*)$. Thus, we may focus on the single variable function ${\psi}_1\colon\, \pi(c_1, c)\to\mathbb{R}$, defined as 
\[\psi_1(a)=\psi(a, a^*).\]
To design an efficient algorithm, we expect some good properties on $\psi_1$. 
Unlike the eccentricity function $\ecc$, function $\psi_1$ is not convex on $\pi(c_1, c)$ (see Figure~\ref{figure:envelopes_and_quasiconvex}(b)). Fortunately, it is quasiconvex. Moreover, for any interval $[c_1, x^*]$ with $c_1\leqslant x^*\leqslant c$, if there is no more than one point at which $\psi_1$ attains the minimum, then $\psi_1$ is strictly quasiconvex on $[c_1, x^*]$ 
(see Lemma~\ref{lemma:quasiconvex}).

\begin{lemma}[strict quasiconvexity]\label{lemma:quasiconvex}
For $a_1< a_2< a_3$, the following statements hold.
\begin{itemize}
\item $\psi_1(a_3)<\psi_1(a_2)$ implies $\psi_1(a_2)<\psi_1(a_1)$;
\item $\psi_1(a_1)<\psi_1(a_2)$ implies $\psi_1(a_2)<\psi_1(a_3)$.
\end{itemize}
\end{lemma}

\begin{proof}
We prove only the statement that $\psi_1(a_3)<\psi_1(a_2)$ implies $\psi_1(a_2)<\psi_1(a_1)$. The other statement can be proved in a similar way. With the assumption that $\psi_1(a_3)<\psi_1(a_2)$, we have
\[(1-\rho)E_L(a_3)+\rho(E_R(a_3)+E_L(a^*_3))<(1-\rho)E_L(a_2)+\rho(E_R(a_2)+E_L(a^*_2)),\]
and therefore
\begin{eqnarray*}
\frac{1-\rho}{\rho}&<&\frac{E_R(a_2)-E_R(a_3)+E_L(a^*_2)-E_L(a^*_3)}{E_L(a_3)-E_L(a_2)}\\
&=&\frac{E_R(a_2)-E_R(a_3)}{E_L(a_3)-E_L(a_2)}+\frac{E_L(a^*_2)-E_L(a^*_3)}{E_R(a^*_3)-E_R(a^*_2)}\\
&=&\frac{(E_R(a_2)-E_R(a_3))/d(a_2, a_3)}{(E_L(a_3)-E_L(a_2))/d(a_2, a_3)}+\frac{(E_L(a^*_2)-E_L(a^*_3))/d(a^*_2, a^*_3)}{(E_R(a^*_3)-E_R(a^*_2))/d(a^*_2, a^*_3)}\\
&\underset{Lemma~\ref{lemma:continuity}}{\leqslant}&\frac{(E_R(a_1)-E_R(a_2))/d(a_1, a_2)}{(E_L(a_2)-E_L(a_1))/d(a_1, a_2)}+\frac{(E_L(a^*_1)-E_L(a^*_2))/d(a^*_1, a^*_2)}{(E_R(a^*_2)-E_R(a^*_1))/d(a^*_1, a^*_2)}\\
&=&\frac{E_R(a_1)-E_R(a_2)+E_L(a^*_1)-E_L(a^*_2)}{E_L(a_2)-E_L(a_1)}.
\end{eqnarray*}
Thus, $\psi_1(a_2)<\psi_1(a_1)$.\qed
\end{proof}
We note here that Lemma~\ref{lemma:quasiconvex} holds for function $\psi_2\colon\, \pi(c, c_2)\to\mathbb{R}$ in a symmetric manner, where 
\[\psi_2(a)=\psi(a^*, a).\]
This property will be used in designing our algorithm, and its proof is similar to that of Lemma~\ref{lemma:quasiconvex}.

\section{A linear time algorithm}\label{section:linear}

The bottleneck on the time complexity of the algorithm of Bhattacharya et al. is the computation of $E_L$ and $E_R$. Fortunately, due to the strict quasiconvexity of $\psi_1$ and the piecewise linearity of $E_L$ and $E_R$, one can apply the strategy of prune-and-search~\cite{Megi83,Tami04} to obtain an optimal solution in linear time. The quasiconvexity of a function $f$ implies that a local minimum of $f$ is the global minimum of $f$, and the idea of the prune-and-search algorithm is to search the local minimum over an interval $[\lambda_1, \lambda_2]$, which is guaranteed to contain the solution. In the search procedure, $[\lambda_1, \lambda_2]$ is recursively reduced to a subinterval, and once it is reduced, the size of the instance can also be pruned with a fixed proportion.
In more detail, the following steps are executed in each recursive call:
\begin{enumerate}
\item\label{step:choose_point}  Choose a point $t$ in $[\lambda_1, \lambda_2]$ appropriately. Initially, $[\lambda_1, \lambda_2]=[c_1, c]$. 
\item\label{step:search_direction} Determine whether $t<b_1$, $t=b_1$, or $t>b_1$.
\item\label{step:discard}  Depending on the result of step~\ref{step:search_direction}, update $[\lambda_1, \lambda_2]$, and discard a subset of vertices without affecting the local optimality of $b_1$ over the updated interval.
\end{enumerate}
During the execution, we claim that the following invariant is maintained. Recall that for any point $a$ in $\pi(c_1, c)$, $a^*$ is a point in $\pi(c, c_2)$ such that $E_L(a)=E_R(a^*)$. 

\begin{claim}[invariant]
In each recursive call, the $6$-tuple $(U_1, U_2, U_3, U_4, \lambda_1, \lambda_2)$ satisfies 
\begin{itemize}
\item $b_1\in[\lambda_1, \lambda_2]$;
\item for each $a\in [\lambda_1, \lambda_2]$, $E_L(a)=\max_{v\in U_1}f^+_v(a)$, $E_R(a)=\max_{v\in U_2}f^-_v(a)$, $E_L(a^*)=\max_{v\in U_3}f^+_v(a^*)$, and $E_R(a^*)=\max_{v\in U_4}f^-_v(a^*)$.
\end{itemize}
\end{claim}

\noindent The invariant was maintained by the procedure $\proc{ReduceInstance1}$, as shown later in Section~\ref{subsection:analysis}.
We note here that to guarantee the efficiency, the steps above have to be repeated in a symmetric manner ($\proc{ReduceInstance2}$). The reason will be clear after the elaboration below. To ease the presentation, we assume that the failure probability $\rho$ can be accessed globally by each procedure. 
The details of steps~\ref{step:choose_point} and~\ref{step:discard} are given in Section~\ref{subsection:step1and3}, that of step~\ref{step:search_direction} is given in Section~\ref{subsection:step2}, and the analysis of the algorithm is given in Section~\ref{subsection:analysis}.

\subsection{Guaranteeing the discarded proportion of vertices}\label{subsection:step1and3}

\begin{figure}[t]
\centering
\includegraphics[width=4.7 in]{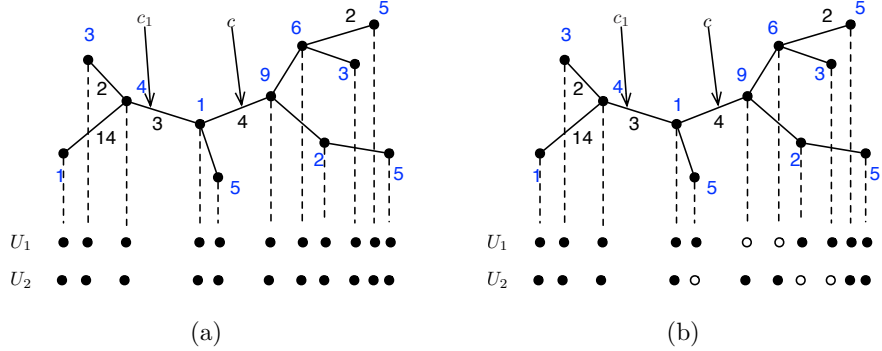}
\caption{(a) The instance before the execution of $\proc{ReduceInstance1}(V, V, V, V, c_1, c)$. In the execution, every two vertices in $U_i$, from left to right, is paired in the partition $\Pi(U_i)$, $\{t^+_{uv}\colon\, \{u, v\}\in\Pi(U_1)\}\cup\{t^-_{uv}\colon\, \{u, v\}\in\Pi(U_2)\}=\{3, -2, 49/4, 29/4, 5/2\}\cup\{-3, 2/3, 39/4, 25/4, 21/2\}$, and $t$ is chosen as $37/8$ (the mean of 3 and $25/4$).  (b) The instance after the execution of $\proc{ReduceInstance1}(V, V, V, V, c_1, c)$. Hollow circles denote the discarded elements from $U_1\cup U_2$. Notice that elements in $U_1$ and $U_2$ need not be sorted.} \label{figure:reduceInstance1}
\end{figure}

The proportion of vertices discarded at step~\ref{step:discard} depends essentially on how the point $t$ at step~\ref{step:choose_point} is chosen. According to the piecewise linearity of $E_L$ and $E_R$, 
the discarded proportion can be guaranteed based on the following simple property.
\begin{property}\label{property:prune_and_search}
Consider two linear functions $f_1(x)=a_1x+b_1$ and $f_2(x)=a_2x+b_2$ with $a_1>a_2$. We have $f_1(x)\leqslant f_2(x)$ if and only if $x\leqslant \frac{b_2-b_1}{a_1-a_2}$.
\end{property}
Similar to the idea in~\cite{Megi83,Tami04}, for $i\in\{1, 2\}$, we arbitrarily partition $U_i$ into $\lfloor|U_i|/2\rfloor$ pairs of vertices, and a single one if $|U_i|$ is odd, and let this partition be denoted by $\Pi(U_i)$. Moreover, let $t^+_{uv}=\varsigma(f^+_u, f^+_v)$ and $t^-_{uv}=\varsigma(f^-_u, f^-_v)$, where $\varsigma(f_1, f_2)$ denotes the point at which $f_1$ and $f_2$ intersect. If $t_m$,  the median of $\{t^+_{uv}\colon\, \{u, v\}\in\Pi(U_1)\}\cup\{t^-_{uv}\colon\, \{u, v\}\in\Pi(U_2)\}$,  is chosen at step~\ref{step:choose_point}, then after step~\ref{step:search_direction}, $\lceil(\lfloor|U_1|/2\rfloor+\lfloor|U_2|/2\rfloor)/2\rceil$ vertices can be discarded without affecting the optimality of $b_1$ in either $[\lambda_1, \min\{t_m, \lambda_2\}]$ or $[\max\{t_m, \lambda_1\}, \lambda_2]$. We implement the above mentioned steps by $\proc{PairPartitionMedian}$ and $\proc{DiscardVertices}$, where $\proc{PairPartitionMedian}$ pairs vertices in $U_1$ and $U_2$, respectively, and choosing the median $t_m$; based on the result and the updated interval containing the solution, i.e., $[\lambda_1, \lambda_2]$, $\proc{DiscardVertices}$ does the process of discarding vertices. An illustration is given in Figure~\ref{figure:reduceInstance1}.

\begin{codebox}
\Procname{$\proc{PairPartitionMedian}(U_1, U_2)$}
\li	$X_1\gets \Pi(U_1)$
\li	$X_2\gets\Pi(U_2)$
\li	$t_m\gets$ the median of $\{t^+_{uv}\colon\, \{u, v\}\in X_1\}\cup \{t^-_{uv}\colon\, \{u, v\}\in X_2\}$
\li	\Return $(t_m, X_1, X_2)$
\end{codebox}

\begin{codebox}
\Procname{$\proc{DiscardVertices}(X_1, X_2, t, \lambda_1, \lambda_2)$}
\li	$U_1\gets \emptyset$
\li	$U_2\gets\emptyset$
\li	\If $t\leqslant\lambda_1$ 	\>\>\>\>\>\Comment $[\lambda_1, \lambda_2]$ is updated, so either $t\leqslant \lambda_1$ or $t\geqslant \lambda_2$
\li		\Then \For each element $x$ in $X_1$
\li				\Do \If $x=\{u, v\}$ and $\varsigma(f^+_u, f^+_v)\leqslant t$
\li						\Then $U_1\gets U_1\cup\{u\colon\, f^+_u(t+\epsilon)>f^+_v(t+\epsilon)\}$
\li						\Else $U_1\gets U_1\cup \{x\}$
					\End
				\End
\li			\For each element $x$ in $X_2$
\li				\Do \If $x=\{u, v\}$ and $\varsigma(f^-_u, f^-_v)\leqslant t$
\li						\Then $U_2\gets U_2\cup\{u\colon\, f^-_u(t+\epsilon)>f^-_v(t+\epsilon)\}$
\li						\Else $U_2\gets U_2\cup \{x\}$
					\End
				\End
\li		\Else 	\For each element $x$ in $X_1$
\li					\Do \If $x=\{u, v\}$ and $\varsigma(f^+_u, f^+_v)\geqslant t$
\li							\Then $U_1\gets U_1\cup\{u\colon\, f^+_u(t-\epsilon)>f^+_v(t-\epsilon)\}$
\li							\Else $U_1\gets U_1\cup \{x\}$
						\End
				\End
\li				\For each element $x$ in $X_2$
\li						\Do \If $x=\{u, v\}$ and $\varsigma(f^-_u, f^-_v)\geqslant t$
\li								\Then $U_2\gets U_2\cup\{u\colon\, f^-_u(t-\epsilon)>f^-_v(t-\epsilon)\}$
\li								\Else $U_2\gets U_2\cup \{x\}$
							\End
						\End
				\End
			\End
		\End
\li	\Return $(U_1, U_2)$
\end{codebox}

\begin{remark}
We use $\epsilon$ to denote a sufficiently small value such that the following relation holds. For two linear function $f$ and $g$,
\[f(t+\epsilon)< g(t+\epsilon) \Leftrightarrow a_1< a_2, \text{ or } a_1=a_2 \text{ and }b_1<b_2,\]
where $f(t+\epsilon)=a_1+b_1\epsilon$ and $g(t+\epsilon)=a_2+b_2\epsilon$.
As a result, no exact value has to be specified for $\epsilon$ in $\proc{DiscardVertices}$. This convention is also adopted in Section~\ref{subsection:step2}.
\end{remark}

\subsection{Evaluating $\psi_1$}\label{subsection:step2} 

Step~\ref{step:search_direction} can be done via evaluating $\psi_1(t_m)$ and $\psi_1(t_m-\epsilon)$. Due to the quasiconvexity of $\psi_1$ (Lemma~\ref{lemma:quasiconvex}), we have that 
\begin{itemize}
\item if $\psi_1(t_m-\epsilon)>\psi_1(t_m)$, then $b_1> t_m$;
\item if $\psi_1(t_m-\epsilon)\leqslant\psi_1(t_m)$, then there exists $b_1\leqslant t_m$.
\end{itemize}
Recall that $\psi_1(a)=(1-\rho)E_L(a)+\rho(E_R(a)+E_L(a^*))$, where $E_R(a^*)=E_L(a)$. For a point $a$ on $\pi(c_1, c)$, the evaluation of $\psi_1(a)$ can be done via computing $E_L(a)$, $E_R(a)$, and $E_L(a^*)$. According to the claim of invariant, for $a\in[\lambda_1, \lambda_2]$, we have 
\begin{equation*}\label{eq:E_L_E_R}
E_L(a)=\max\{f^+_v(a)\colon\, v\in U_1\}\mbox{ and } E_R(a)=\max\{f^-_v(a)\colon\, v\in U_2\}.
\end{equation*}
Clearly, $E_L(a)$ and $E_R(a)$ can be computed in time linear to $|U_1|$ and $|U_2|$, respectively. If $a^*$ is given, then $E_L(a^*)$ can also be computed in $O(|U_4|)$ time. It remains to show how $a^*$ is determined. 

\medskip

Given a value $\xi$, since $E_R(a)=\max\{f^-_v(a)\colon\, v\in U_3\}$, the point $a$ on $\pi(c, c_2)$ satisfies $E_R(a)=\xi$ if and only if
\[a=\max\{x\colon\, f^-_v(x)=\xi, v\in U_3\}.\]
Therefore, for $a\in[\lambda_1, \lambda_2]$, determining $a^*$ can be done in $O(|U_3|)$ time, and $E_L(a^*)$ can then be computed in $O(|U_4|)$ time. Formally, the procedure of evaluating $\psi_1$ at point $a$ in $[\lambda_1, \lambda_2]$ is given as $\proc{Evaluate1}(a, U_1, U_2, U_3, U_4)$.

\medskip

\begin{codebox}
\Procname{$\proc{Evaluate1}(a, U_1, U_2, U_3, U_4)$}
\li	$\xi_1\gets \max\{f^+_v(a)\colon\, v\in U_1\}$
\li	$\xi_2\gets\max\{f^-_v(a)\colon\, v\in U_2\}$
\li	$a^*\gets\max\{x\colon\, f^-_v(x)=\xi_1, v\in U_4\}$
\li	$\xi_3\gets\max\{f^-_v(a^*)\colon\, v\in U_3\}$
\li	\Return $(1-\rho)\xi_1+\rho(\xi_2+\xi_3)$
\end{codebox}
The evaluation of $\psi_2$ at a given point $a$ can be done symmetrically, and the details are omitted. 

\subsection{The analysis of the algorithm}\label{subsection:analysis}

With the procedures given in Sections~\ref{subsection:step1and3} and~\ref{subsection:step2}, we may implement the idea given in the beginning of Section~\ref{section:linear}, which recursively reduces the size of the problem instance. The procedure is given as $\proc{ReduceInstance1}$.

\medskip

\begin{codebox}
\Procname{$\proc{ReduceInstance1}(U_1, U_2, U_3, U_4, \lambda_1, \lambda_2)$}
\li	$(t, X_1, X_2)\gets\proc{PairPartitionMedian}(U_1, U_2)$ \label{line:choose_median}
\li	$\xi_1\gets\proc{Evaluate1}(t-\epsilon, U_1, U_2, U_3, U_4)$
\li	$\xi_2\gets\proc{Evaluate1}(t, U_1, U_2, U_3, U_4)$
\li	\If $\xi_1<\xi_2$
\li		\Then	$\lambda_2\gets\min\{t, \lambda_2\}$
\li		\Else	$\lambda_1\gets\max\{t, \lambda_1\}$
	\End
\li	$(U_1, U_2)\gets\proc{DiscardVertices}(X_1, X_2, t, \lambda_1, \lambda_2)$ \label{line:discard_critical_vertices}
\li	\Return $(U_1, U_2, \lambda_1, \lambda_2)$
\end{codebox}

\medskip

\noindent An example is given in Figure~\ref{figure:reduceInstance1}. Since $t$ is chosen as the median of $\{t^+_{uv}\colon\, \{u, v\}\in \Pi(U_1)\}\cup \{t^-_{uv}\colon\, \{u, v\}\in \Pi(U_2)\}$ (line~\ref{line:choose_median}), it can be derived that a fixed proportion of vertices in $U_1\cup U_2$ are discarded after the execution of $\proc{DiscardVertices}$ (line~\ref{line:discard_critical_vertices}). Moreover, together with Property~\ref{property:prune_and_search}, it can be easily derived that for $a\in[\lambda_1, \lambda_2]$, $\psi_1(a)$ remains unchanged. We summarize these properties in Lemma~\ref{lemma:discarded_portion}.

\medskip

\begin{lemma}\label{lemma:discarded_portion}
 After the execution of $\proc{ReduceInstance1}(U_1, U_2, U_3, U_4, \lambda_1, \lambda_2)$, at least $\left\lceil\frac{\lfloor n_1/2\rfloor+\lfloor n_2/2\rfloor}{2}\right\rceil$ vertices of $U_1\cup U_2$ are discarded, where $n_1=|U_1|$ and $n_2=|U_2|$. Moreover, the procedure maintains the invariant, the 6-tuple $(U_1, U_2, U_3, U_4, \lambda_1, \lambda_2)$, in which 
\begin{itemize} 
\item $b_1\in[\lambda_1, \lambda_2]$; 
\item for $a\in[\lambda_1, \lambda_2]$, $E_L(a)=\max\{f_v^+(a)\colon\, v\in U_1\}$, $E_R(a)=\max\{f_v^-(a)\colon\, v\in U_2\}$, $E_L(a^*)=\max\{f_v^+(a^*)\colon\, v\in U_3\}$, $E_R(a^*)=\max\{f_v^-(a^*)\colon\, v\in U_4\}$.
\end{itemize}
\end{lemma}

\medskip

\noindent Notice that after the execution of $\proc{ReduceInstance1}(U_1, U_2, U_3, U_4, \lambda_1, \lambda_2)$, only a proportion of $U_1\cup U_2$ is discarded. To ensure that a fixed proportion of the instance is pruned, in a symmetric manner, one can apply a similar procedure to discard a proportion of $U_3\cup U_4$, as noted in the beginning of Section~\ref{section:linear}. We name the corresponding procedure as $\proc{ReduceInstance2}$. 

\medskip

\begin{remark}
$\proc{ReduceInstance2}$ is the same as $\proc{ReduceInstance1}$ except it evaluates $\psi_2$ instead of $\psi_1$ (see lines~2 and~3 of $\proc{ReduceInstance1}$).
\end{remark}

\medskip

With  $\proc{ReduceInstance1}$ and  $\proc{ReduceInstance2}$, one may recursively reduce the size of the instance until the instance is small enough. For small instance with both $|U_1\cup U_2|\leqslant 3$ and $|U_3\cup U_4|\leqslant 3$, we compute the solution by evaluating $\psi_1$ at all bending points in $[\lambda_1, \lambda_2]$ since $\psi_1$ is piecewise linear. We denote this procedure by $\proc{SmallInstance}$.

\medskip

The integration is given as $\proc{PruneAndSearch}$, and the procedure for computing a weighted backup 2-center is given as $\proc{Backup2Center}$. The correctness and time complexity are analyzed in Theorem~\ref{theorem:correctness_and_time}.

\begin{codebox}
\Procname{$\proc{PruneAndSearch}(U_1, U_2, U_3, U_4, \lambda_1, \lambda_2, \lambda_3, \lambda_4)$}
\li	\If $|U_1\cup U_2|\leqslant 3$ and $|U_3\cup U_4|\leqslant3$
\li		\Then	\Return $\proc{SmallInstance}(U_1, U_2, U_3, U_4, \lambda_1, \lambda_2, \lambda_3, \lambda_4)$		
\li		\Else		$(U'_1, U'_2, \lambda'_1, \lambda'_2)\gets\proc{ReduceInstance1}(U_1, U_2, U_3, U_4, \lambda_1, \lambda_2)$ \label{line:reducebegin}
\li				$(U'_3, U'_4, \lambda'_3, \lambda'_4)\gets\proc{ReduceInstance2}(U_1, U_2, U_3, U_4, \lambda_3, \lambda_4)$
\li				\Return  $\proc{PruneAndSearch}(U'_1, U'_2, U'_3, U'_4,  \lambda'_1, \lambda'_2, \lambda'_3, \lambda'_4)$ \label{line:reduceend}
	\End
\end{codebox}

\medskip

\begin{codebox}
\Procname{$\proc{Backup2Center}(T, \rho)$}
\li	$c\gets$ center of $T$
\li	$(c_1, c_2)\gets$ 2-center of $T$ with the property in Proposition~\ref{proposition:two_center_equal} 
\li	\Return $\proc{PruneAndSearch}(V, V, V, V, c_1, c, c, c_2)$
\zi \>\>\Comment $V$ is the vertex set of $T$
\end{codebox}

\medskip

\begin{theorem}\label{theorem:correctness_and_time}
The weighted backup 2-center problem on a tree can be solved in linear time by $\proc{Backup2Center}$.
\end{theorem}

\begin{proof}
Let $T=(V, E)$ with $|V|=n$. By Lemma~\ref{lemma:solution_location}, $b_1\in[c_1, c]$ and  $b_2\in [c, c_2]$, and thus $\lambda_1$, $\lambda_2$, $\lambda_3$, and $\lambda_4$ are initialized accordingly as in $\proc{Backup2Center}$. Besides, by definition, we have $E_L(a)=\max\{f^+_v(a)\colon\, v\in V\}$ and $E_R(a)=\max\{f^-_v(a)\colon\, v\in V\}$. With Lemma~\ref{lemma:discarded_portion}, the initialization of $U_1=V$ and $U_2=V$ guarantees $\psi_1(a)$, for $a\in [\lambda_1, \lambda_2]$, is computed correctly. Similar arguments hold for the initialization of $U_3=V$ and $U_4=V$. 

For the time complexity, both the center and the 2-center can be computed in $O(n)$ time~\cite{BeBh06,Megi83}. For $\proc{PruneAndSearch}$, let $n_1=|U_1|$, $n_2=|U_2|$, $n_3=|U_3|$, and $n_4=|U_4|$. Lines~\ref{line:reducebegin}--\ref{line:reduceend} of $\proc{PruneAndSearch}$ are executed if either $|U_1\cup U_2|\geqslant 4$ or $|U_3\cup U_4|\geqslant 4$. Together with Lemma~\ref{lemma:discarded_portion}, it can be derived that
\begin{equation}\label{eq:proportion}
\left\lceil\frac{\lfloor n_1/2\rfloor+\lfloor n_2/2\rfloor}{2}\right\rceil+\left\lceil\frac{\lfloor n_3/2\rfloor+\lfloor n_4/2\rfloor}{2}\right\rceil\geqslant \frac18(n_1+n_2+n_3+n_4),
\end{equation}
where the equality holds when $n_1=3$, $n_2=3$, $n_3=1$, and $n_4=1$.
As a result,  let the execution time of $\proc{PruneAndSearch}(V, V, V, V, c_1, c, c, c_2)$ be $T(N)$, where $N= |V|+ |V|+|V|+|V|=4n$. It follows from~(\ref{eq:proportion}) that
\[T(N)\leqslant T\left(\frac{7N}{8}\right)+O(N).\]
Therefore, $T(N)=O(N)=O(n)$. \qed
\end{proof}

\section{Concluding remarks}\label{section:conclusion}

In this paper, we propose a linear time algorithm to solve the weighted backup 2-center problem on a tree, which is asymptotically optimal. Based on the observations given by Bhattacharya et al.~\cite{BhDe14}, ``good properties'' of the  objective function are further derived. With these properties, the strategy of prune-and-search can be applied to solve this problem. For future research, the hardness of the backup $p$-center problem on trees is still unknown, even for the unweighted case. It worth investigation on this direction.

\subsubsection*{Acknowledgements}
The author would also like to thank Professor Kun-Mao Chao, Ming-Wei Shao, and Jhih-Heng Huang for fruitful discussions. Hung-Lung Wang was supported in part by MOST grant 103-2221-E-141-004, from the Ministry of Science and Technology, Taiwan.

\end{document}